\documentclass[12pt]{article}
\usepackage{amsmath}
\usepackage{amssymb}
\usepackage{amsthm}
\usepackage{exscale}
\usepackage[mathscr]{eucal}
\usepackage{float}
\usepackage{bm}

\newtheorem{theorem}{Theorem}[section]
\newtheorem{lemma}[theorem]{Lemma}
\newtheorem{corollary}[theorem]{Corollary}

\newtheorem{remark}[theorem]{Remark}

\numberwithin{equation}{section}
\begin{document}

\title{Logarithmic Quantile Estimation for Rank Statistics}

\author{Manfred Denker, Lucia Tabacu
       \\
        \\
				Mathematics Department, Department of Statistics\\
        Pennsylvania State University
        }
\date{December 10, 2012}

\maketitle
\begin{abstract} We prove an almost sure weak limit theorem for simple linear
        rank statistics for samples with continuous distributions functions. As a
        corollary the result extends to samples with ties, and the vector
        version of an a.s. central limit theorem for vectors of linear rank
        statistics. Moreover, we derive such a weak convergence result for some  quadratic
        forms.  These results are then applied
        to quantile estimation, and  to hypothesis testing for
        nonparametric statistical designs, here demonstrated by the c-sample
        problem, where the samples may be dependent. In general, the method is known to be
        comparable to 
        the bootstrap and other nonparametric  methods (\cite{THA, FRI}) and
        we  confirm  this finding for the c-sample problem. 
\end{abstract}  

\section{Almost sure central limit theorem for statistical functionals}\label{sec:1}
Almost sure central limit theorems (ASCLT) are based on a new type of averaging procedures for sums of  i.i.d. random variables to obtain their asymptotic distribution. Its application to statistical functionals  is a fairly new subject. 
It has been recently observed by \cite{THA} and \cite{FRI} that the logarithmic averages of data may be used for quantile estimation in practice. The philosophy behind this procedure resembles somehow to Efron's bootstrap method (\cite{EFR}) but resampling is not needed in the almost sure method. In this note we extend the applicability of the almost sure quantile estimation to rank statistics. We shall call it logarithmic quantile estimation (LQE).\\
We develop the method for general rank models as defined in \cite{BR1}, but
adapted to the needs of the almost sure concept. Let $X_1,...,X_N$ be a sample of random vectors and for each $n\le N$, let $T_n$ be a statistic based on $X_1,...,X_n$. The logarithmic average of the sequence $T_n$ has the form
\begin{equation*}
\widehat{G}_{N}(t)=\frac{1}{C_{N}}\sum_{n=1}^{N}\frac{1}{n}\mathbb I(T_n\le t),   
\end{equation*}
where $C_N$ is chosen to make $\widehat G_N$ an empirical distribution
function, and where $\mathbb I_C$ denotes the indicator function of the set $C$.
Note that $C_N\asymp \ln N$, which is responsible for the name log
averaging. In fact, the usage of $\widehat G_N$  is as for the classical
empirical distribution functions. We do not state the details like the
corresponding Glivenko-Cantelli theorem or the almost sure convergence of the
LQE-quantiles. These results are easy to prove and left as an exercise.
Then the empirical $\alpha-$quantile of $\widehat{G}_{N}$ can be used in
hypothesis testing, for example a typical rejection region may look like
$\{X\in \mathbb{R}^{dN}: |T_N| \ge z_{\alpha} \}$ with
$\widehat{G}(z_{\alpha})=\alpha$. The details are as well left as an exercise.\\
Let us give a brief overview of some important results on almost sure central limit theorems, which is the base for the validity of good test procedures. The concept of almost sure central limit theorem has its origins in the work of \cite{BRO}, \cite{SCH} and \cite{LAC}. Since then many important results have been obtained for independent and dependent random variables, for random vectors, stochastic processes as well as their almost sure functional versions.\\
The simplest form of the almost sure central limit theorem is 
\begin{equation*}
\lim_{n\rightarrow \infty}\frac{1}{\ln n}\sum_{k=1}^{n}\frac{1}{k}\mathbb I(\frac{S_{k}(\omega)}{\sqrt{k}}<x)=\Phi(x) \text{ for almost all }  \omega \in \Omega,
\end{equation*}
and all numbers $x$, where $X_1(\omega), X_2(\omega),...$ are independent and identically distributed random variables defined on a probability space $(\Omega, \mathcal{B},P )$ with $E(X_1)=0, E(X_1^{2})=1$, $S_k(\omega)=X_1(\omega)+...+X_k(\omega)$, $I$ is the indicator function and $\Phi$ is the standard normal distribution function.\\  
\cite{BER} obtained general results that extend weak limit theorems for independent random variables to their almost sure versions. Applications of their results are the almost sure central limit theorems for partial sums, maxima of partial sums, extremes, empirical distribution functions, U-statistics (see also \cite{HOL}), Darling-Erd\"os type limit theorems. \cite{BED} proved a general almost sure central limit theorem for independent, not identically distributed random variables. \cite{PEL} considered the almost sure central limit theorem for weakly dependent random variables, \cite{LI1} obtained the almost sure limit theorem for sums of random vectors, \cite{CHU} treated the case of Pearson statistic.\\    
\cite{THA} proved an almost sure central limit theorem for the two-sample
linear rank statistics and developed a hypothesis testing procedure based on
the almost sure convergence. He applies his method to problems like testing
for the mean in the parametric one-sample problem, testing for the equality of
the means in the parametric two-sample problem, and for the nonparametric
Behrens-Fisher problem. As a result he showed that the LQE method is better
than bootstrap and almost as good as the t-test for the two sample
problem. For the nonparametric two sample Behrens-Fisher problem he compared  the LQE
method with the methods in \cite{BAB}, \cite{REI} and \cite{BR4}. It is shown
that the LQE-method performs stably over various distributions, is comparable
to the other methods and often preferable. Later, \cite{DEF} obtained the
almost sure version of Cram\'er's theorem.  Using this result, \cite{FRI}
showed the almost sure central limit theorem for the population correlation
and applied the almost sure version of Cram\'er's theorem to obtain confidence
intervals for the population correlation coefficient. It turns out that the
LQE method is superior to bootstrap for this statistics. \\
All these results show that the LQE method has to be developed further, in
particular, for nonparametric designs when bootstrap methods are hardly
possible to apply. \cite{STE} has obtained a bootstrap method for simple
rank statistics using the von Mises method, and assuming complete independence. Here we are interested in the general result, when samples are not
identically distributed and hence resampling becomes doubtful. The LQE method
does not have this restriction, and we show that it provides good results. As a special design we
chose the c-sample problem when the samples are independent or not: we show
that the test based on the LQE method provides better coverage probability
than the classical Kruskal-Wallis test. We also show that in the dependent
situation the LQE test has a satisfying performance. For other designs we got similar
results; this will be published elsewhere. Besides this advantage, the
LQE-method estimates quantiles directly from the data, not using the
asymptotic distribution, hence it also does not use any estimation of unknown
variances or covariances or eigenvalues of covariance matrices. It is also
applicable when asymptotic covariance matrices become degenerate.  
\\
The article is organized as follows. In Section \ref{sec:2} we introduce the
general model for simple linear rank statistics and state our result on the
almost sure central limit theorem for simple linear rank statistics. In
Section \ref{sec:3} we treat the c-sample problem when the samples are
independent as well as dependent. The interesting point here is that we derive
an almost sure central limit theorem for Kruskal-Wallis statistic even in the
case when the single observations have non-independent coordinates. To our
knowledge no test treats this general case without any further assumption.  This is then applied to hypothesis testing. Section \ref{sec:4} contains the results of our simulation study. We compute the empirical logarithmic quantiles for three independent samples when the Kruskal-Wallis statistic is used and we calculate the power and type I error for three dimensional vectors with dependent coordinates. In Section \ref{sec:5} we provide proofs of all the auxiliary lemmas that we used to prove the main theorem.

\section{Almost sure central limit theorem for rank statistics}\label{sec:2}

We begin this section stating the model assumptions. Since they are different
from the standard literature as  in \cite{BR1} and subsequently in \cite{BR2},
\cite{BR3}, \cite{AK1}, and  \cite{AK2} to name a few, we need to state the
notation in as much as it differs from those references. Note that in the general model for each $n$, an array of independent random vectors $\mathbf{X}_{i}(n)=(X_{i1}(n),...,X_{im_{i}(n)}(n))$, with $i=1, 2,...,n$ and $n\in \mathbb{N}$ was defined on, may be, different probability spaces.\\
Since we are heading for an almost sure type result we need to consider the
model on a common probability space. That is why our model requires a sequence of independent random vectors $\mathbf{X}_{i}=(X_{i1},...,X_{im_{i}})$, $i=1, 2,...$ with continuous marginal distributions
\begin{equation*}
F_{ij}(x)=P(X_{ij}\leq x), x\in \mathbb{R}, j=1,...,m_i.
\end{equation*}
Note that in \cite{MUN} this condition of continuity was shown to be
unnecessary. Likewise the theorem below holds as well when ties are present,
as it is well known that one can replace ranks by midranks. For simplicity, we keep the commonly used assumption of having no ties. \\
Note that we allow dependence of the coordinates of the random vectors, each
vector may have a different dependence structure. This general approach
excludes the use of the bootstrap method.  As it is well known (see \cite{BR1}
  and subsequently \cite{BR2}, for example) this relaxation of the classical assumptions can be used for a large class of designs, for example repeated measure designs or time series observations. In order to set up the notation for rank statistics we introduce the following notations:\\ 
For $n\ge 1$, let $N(n)=\sum_{i=1}^{n}m_{i}$ denote the number of observations involved in the vectors $\mathbf{X}_1,...,\mathbf{X}_n$
and let $\lambda_{ij}^{(n)} \;(1\leq j\leq m_{i}, i\ge 1)$ be (known) regression constants which are assumed to satisfy 
\begin{equation}\label{eq:3.2}
\sideset{}{}\max_{1\le i\le n,\;1\le j\le m_{i}}|\lambda_{ij}^{(n)}|=1
\end{equation}
Define 
\begin{align}\label{eq:3.3}
H^{(i)}(x) &=\sum_{j=1}^{m_{i}}F_{ij}(x),&
\widehat{H}^{(i)}(x) &=\sum_{j=1}^{m_{i}}\mathbb I(X_{ij}\leq x)\\
\label{eq:3.4} F^{(i,n)}(x) &=\sum_{j=1}^{m_{i}}\lambda_{ij}^{(n)} F_{ij}(x),&
\widehat{F}^{(i,n)}(x) &=\sum_{j=1}^{m_{i}}\lambda_{ij}^{(n)}\mathbb I(X_{ij}\leq x)\\
\label{eq:3.5} H_{n}(x)&=\frac{1}{N(n)}\sum_{i=1}^{n}H^{(i)}(x),&
\widehat{H}_{n}(x)&=\frac{1}{N(n)}\sum_{i=1}^{n}\widehat{H}^{(i)}(x)\\
\label{eq:3.6} F_{n}(x)&=\frac{1}{N(n)}\sum_{i=1}^{n}F^{(i,n)}(x),&
\widehat{F}_{n}(x)&=\frac{1}{N(n)}\sum_{i=1}^{n}\widehat{F}^{(i,n)}(x).
\end{align}
\\
The simple linear rank statistic that we are interested in is defined by
\begin{equation}\label{eq:3.7}
L_{n}(J)=\int_{-\infty}^{\infty}J\left(\frac{N(n)}{N(n)+1}\widehat{H}_{n}\right)d\widehat{F}_{n}=
\frac{1}{N(n)}\sum_{i=1}^{n}\sum_{j=1}^{m_{i}}\lambda_{ij}^{(n)}J\left(\frac{R_{ij}(n)}{N(n)+1}\right),
\end{equation}
where $R_{ij}(n)$ denotes the rank of $X_{ij}$ among all random variables $\{X_{kl}: 1\leq k\leq n, 1\leq l\leq m_{k}\}$
and $J:(0, 1)\rightarrow \mathbb{R}$ denotes an (absolutely continuous) score function. Let
\begin{equation}\label{eq:3.8}
T_{n}(J)=L_{n}(J)-\int_{-\infty}^{\infty}J(H_{n})dF_{n},
\end{equation}
\begin{equation}\label{eq:3.9}
s_{n}^{2}(J)=N(n)^{2}E(T_{n}(J)^{2}),
\end{equation}
\begin{equation}\label{eq:3.10}
B_{n}(J)=\int_{-\infty}^{\infty} J(H_{n})d(\widehat{F}_{n}-F_{n})+\int_{-\infty}^{\infty} J'(H_{n})(\widehat{H}_{n}-H_{n})dF_{n},\\
\end{equation}
\begin{equation}\label{eq:3.11}
\sigma_{n}^{2}(J)=N(n)^{2}\text{Var}(B_{n}(J)).
\end{equation}
The asymptotic normality of the linear rank statistics for independent random vectors with varying dimension introduced above was proved in \cite{BR1} (Theorem 3.1).\\
The main result of this note is an almost sure central limit theorem for the statistics defined in (\ref{eq:3.8}).\\
\\
\begin{theorem}\label{theo:3.1} Let $J:(0,1)\rightarrow \mathbb{R}$ be a twice differentiable score function with bounded second derivative and let $\lambda_{ij}^{(n)}$ be regression constants satisfying (\ref{eq:3.2}). Then the rank statistics (\ref{eq:3.8}) satisfies the almost sure central limit theorem, that is 
\begin{equation*}
\lim _{N\rightarrow \infty} \frac{1}{\ln N}\sum_{n=1}^{N}\frac{1}{n}\mathbb I(\frac{N(n)}{\sigma_{n}(J)}T_{n}(J)\le t)=\Phi(t)
\end{equation*}
provided\\
(a) $\sigma_n(J)$ defined in (\ref{eq:3.11}) satisfies, for some $M>0, \gamma >0$\\ 
\begin{equation}\label{eq:3.13}
\frac{\sigma_{m}(J)}{\sigma_{n}(J)}\ge M\left(\frac{m}{n}\right)^{\gamma}, \text{ for } m\ge n  
\end{equation}
(b) For $n_{k}=\min\{j: N(j)\ge k^{2}\}$ one has that 
\begin{equation}\label{eq:3.14}
\sum_{k=1}^{\infty}\left(\frac{\max_{1\leq i\leq n_{k}}m_{i}}{\sigma_{n_{k}}(J)}\right)^{2}<\infty,
\end{equation}
\begin{equation}\label{eq:3.15}
(\log n_k)\frac{\max_{1\leq i\leq n_{k}} m_{i}}{\sigma_{n_{k}}(J)}\to 0,
\end{equation}
and for $n_{k}\le j < n_{k+1}$ and a constant $K$,
\begin{equation}\label{eq:3.16}
\max_{1\le \lambda\le j}m_{\lambda}\le K\max_{1\le i\le n_{k}}m_{i}.
\end{equation}
\end{theorem}

\begin{remark}\label{rem:3.2} {\rm Note that the essential assumptions in the theorem are (a) and (b). (a) is a condition on the growth of the variances of the asymptotically equivalent statistics $B_n$ (see \cite{BR1} for details). This condition is easily verified in many examples, e.g. if $\sigma_{n}^{2}(J)=O(n)$. (b) is in fact a condition on the maximal allowable dimensions of the vectors $X_i$. If all $m_i$=1 then condition (b) is trivially satisfied, the same is true if $\max_{i\ge 1}m_i <\infty.$}
\end{remark}

\begin{proof}
The proof of the theorem follows from a standard decomposition (\cite{BR1}):  Taylor expansion of J around $H_{n}(t)$ and integration by parts yields
\begin{equation*}
\frac{N(n)}{\sigma_{n}(J)}T_{n}(J)=\frac{N(n)}{\sigma_{n}(J)}B_{n}(J)+\frac{N(n)}{\sigma_{n}(J)}C_{1}(n)-\frac{N(n)}{\sigma_{n}(J)}C_{2}(n)+
\frac{N(n)}{\sigma_{n}(J)}C_{3}(n),
\end{equation*}
where
\begin{align}\label{eq:3.18}
     C_{1}(n)&=\int_{-\infty}^{\infty} J'(H_{n})(\widehat{H}_{n}-H_{n})d(\widehat{F}_{n}-F_{n}), &\\
 \label{eq:3.19}    C_{2}(n)&=\frac{1}{N(n)+1}\int_{-\infty}^{\infty} J'(H_{n})\widehat{H}_{n}d\widehat{F}_{n}, &\\
 \label{eq:3.20}    C_{3}(n)&=\frac{1}{2}\int_{-\infty}^{\infty} J''(\theta(H_{n}))\left(\frac{N(n)}{N(n)+1}\widehat{H}_{n}-H_{n}\right)^{2}d\widehat{F}_{n},
\end{align}
and $\theta(H_{n})\in [H_{n}, \frac{N(n)}{N(n)+1}\widehat{H}_{n}]\cup[\frac{N(n)}{N(n)+1}\widehat{H}_{n}, H_{n}]$.\\
\\
By Lemmas \ref{lem:6.2}, \ref{lem:6.3}, \ref{lem:6.5} it follows that
\begin{equation}\label{eq:3.21}
\frac{N(n)}{\sigma_{n}(J)}C_{1}(n)-\frac{N(n)}{\sigma_{n}(J)}C_{2}(n)+\frac{N(n)}{\sigma_{n}(J)}C_{3}(n)\rightarrow 0 \text{ a.s.  when }n\rightarrow \infty. 
\end{equation}
By Lemma 2.2 in \cite{FRI}, Lemma \ref{lem:6.6} and (\ref{eq:3.21}), we obtain the almost sure central limit theorem for the statistics $T_n(J)$
\begin{equation*}
\lim _{n\rightarrow \infty}\frac{1}{\ln n}\sum_{k=1}^{n}\frac{1}{k}\mathbb I(\frac{N(k)}{\sigma_{k}(J)}T_{k}(J)\le t)=\Phi(t).
\end{equation*}
\end{proof}

\noindent The next corollary is a form of the theorem which can be used for hypothesis testing. \\
\\
\begin{corollary}\label{cor:3.3} Assume that in Theorem \ref{theo:3.1}, 
\begin{equation*}
\sigma_{n}^2(J)=a_n^2\sigma^2+o(a_n)   
\end{equation*}
where $\sigma^2>0$ and $a_n$ satisfies (\ref{eq:3.13}) when replacing
$\sigma_{n}(J)$ by $a_n$. \\
Then under (b), the statistics $\frac{N(n)}{a_n}T_{n}(J)$ satisfies the central limit theorem and the almost sure central limit theorem, that is for $t \in \mathbb{R}$
\begin{equation*}
\lim_{n\rightarrow \infty} P(\frac{N(n)}{a_n}T_{n}(J)\le t)=\frac{1}{\sqrt{2\pi \sigma^2}}\int_{-\infty}^{t} e^{-\frac{u^2}{2\sigma^2}}du
\end{equation*}
and 
\begin{equation*}
\lim_{n\rightarrow \infty} \frac{1}{\ln n}\sum_{k=1}^n \frac{1}{k} \mathbb I(\frac{N(k)}{a_k}T_{k}(J)\le t)=\frac{1}{\sqrt{2\pi \sigma^2}}\int_{-\infty}^{t} e^{-\frac{u^2}{2\sigma^2}}du \text{ a.s. }
\end{equation*}
\end{corollary}
\begin{proof} See \cite{BR1} for the first part. The second part is a special case of Theorem \ref{theo:3.1}.
\end{proof}
\noindent The next remark states the properties under which  hypothesis
testing for $L_n(J)$ is possible.\\
\\
\begin{remark}\label{rem:3.4} 
{\rm Let $H_{0}: \frac{N(n)}{\sigma_{n}(J)}\left(\int J(H_n)dF_{n} -
c\right)\to 0$ as $n\to 0$ for some constant c. Then under the null hypothesis
$Q_{n}(J):=\frac{N(n)}{\sigma_{n}(J)}(L_{n}(J)-c)$ is asymptotically normal
and satisfies the almost sure central limit theorem.\\
Let $\hat t_{\gamma}$ denote the empirical $\gamma$-quantile
of the empirical distribution function $\widehat G_N$. Let $\alpha>0$.\\
Then, under the null hypothesis,  
\begin{equation*}
I_{\alpha}^{(N)}=\left[Q_N(J)-\hat{t}_{1-\alpha}^{(N)}, Q_N(J)-\hat{t}_{\alpha}^{(N)} \right]
\end{equation*}
is a random interval with the property that
\begin{equation*}
P(0\in I_{\alpha}^{(N)})\rightarrow 1-2\alpha.
\end{equation*}
If $H_1: \frac{N(n)}{\sigma_{n}(J)}(\int J(H_n)dF_{n}- d) \to 0$  for $d\neq
c$ and if $\frac{N(n)}{\sigma_{n}(J)}\to \infty$, then $|Q_{n}(J)|\to \infty$
a.s.}
It follows that under these conditions the power of the test approaches 1
under the alternative. 

The proof of this fact is left as an exercise. 
\end{remark}

\section{Example: the dependent $c$-sample problem}\label{sec:3}
In this section we consider the problem of testing the equality of the
distributions of $c$ samples that may be dependent. This problem is chosen to
show how  Theorem \ref{theo:3.1} is used to derive LQE-quantiles for quadratic
forms based on vectors of ranks statistics of the form (\ref{eq:3.8}). 
In order to do this we derive the vector form of Theorem \ref{theo:3.1} and
conclude the a.s. convergence of the quadratic form. Note that the
result on distributional convergence is well known, hence we can proceed as in
Remark \ref{rem:3.4} to derive the test. Note that the great advantage of the
LQE-method is that the form of the limiting distribution need not to be known
(like for bootstrap), so there is no need to estimate covariances of unknown
limiting distributions. This makes many problems easier and accessible under
no further assumption on covariances. It also covers cases of degenerate
covariance matrices.

The test statistic that we are using here is the classical Kruskal-Wallis
statistic. In the case of dependent samples the distribution of the test
statistic  and its asymptotic distribution under the null hypothesis are not
known and a statistical decision is not possible to derive from this. In order
to apply the logarithmic quantile estimation, we first need to show that the
Kruskal-Wallis statistic is a particular statistics derived from the simple
linear rank statistic $T_{n}(J)$ as defined in (\ref{eq:3.8}) and that it satisfies an almost sure central limit theorem.
Let $\mathbf{X}_{i}=(X_{i1},...,X_{ic})$ ($1\leq i\leq n$) be independent
random vectors such that the vectors $(X_{1k}, X_{2k},...,X_{nk})'$, for
$k=1,2,...,c$, are possibly dependent random variables  with continuous marginal distribution functions $F_{k}(t)=P(X_{ik} \leq t)$ for $k=1,...,c$. We use the Wilcoxon scores $J(t)=t$ so that $J'(t)=1$ and $\sigma_{n}=\sqrt{n}$. Then with the additional fixed index $l$ that corresponds to the $l$-th sample, definitions (\ref{eq:3.3})-(\ref{eq:3.10}) can be written as
\begin{equation*}
\lambda_{ij}^{(l)}=\begin{cases} 1, j=l\\
                                 0, j\neq l   
																\end{cases}
\end{equation*}  
where $1\le i\le n, 1\le j\le c$,
\begin{eqnarray*}
H_{n}(t) &=&\frac{1}{N(n)}\sum_{i=1}^{n}\sum_{k=1}^{c}F_{k}(t)=\frac{1}{c}\sum_{k=1}^{c}F_{k}(t), 
\hat{H}_{n}(t) =\frac{1}{N(n)}\sum_{i=1}^{n}\sum_{k=1}^{c}\mathbb I(X_{ik}\le t)\\
F_{n}^{(l)}(t) &=&\frac{1}{N(n)}\sum_{i=1}^{n}F_{l}(t)=\frac{1}{c}F_{l}(t),\qquad
\hat{F}_{n}^{(l)}(t) =\frac{1}{N(n)}\sum_{i=1}^{n}\mathbb I(X_{il}\le t)
\end{eqnarray*}
and
\begin{equation*}
T_{n}^{(l)}=L_{n}^{(l)}-\int{H_{n}(t)dF_{n}^{(l)}}=\frac{1}{N(n)(N(n)+1)}\sum_{i=1}^{n}R_{il}-\frac{1}{c^{2}}\sum_{j=1}^{c}\int{F_{j}(t)dF_{l}(t)}
\end{equation*}
\begin{eqnarray*}
B_{n}^{(l)}& = &\int{H_{n}(t)}d(\hat{F}_{n}^{(l)}-F_{n}^{(l)})(t)+\int{(\hat{H}_{n}(t)-H_{n}(t))dF_{n}^{(l)}(t)}=\\
&
& \frac{1}{N(n)c}\sum_{i=1}^{n}\sum_{j=1}^{c}F_{j}(X_{il})+\frac{1}{N(n)c}\sum_{i=1}^{n}\sum_{j=1}^{c}\int{\mathbb
I(X_{ij}\le t)dF_{l}(t)}\\
& & -\frac{2}{c^2}\sum_{j=1}^{c}\int{F_{j}(t)dF_{l}(t)}=\frac{1}{N(n)c}\sum_{i=1}^{n}\left(\sum_{j=1}^{c}F_{j}(X_{il})+ \right.\\
& & \left.+\sum_{j=1}^{c}\int{\mathbb I(X_{ij}\le t)dF_{l}(t)}-2\sum_{j=1}^{c}\int{F_{j}(t)dF_{l}(t)}\right).
\end{eqnarray*}
\\
Define the independent c-dimensional vectors 
\begin{eqnarray}\label{eq:4.6}
&&\xi_{i}=(\xi_{ik})_{1\le k\le c}\notag \\
& &=\left(\sum_{j=1}^{c}F_{j}(X_{ik})+\sum_{j=1}^{c}\int{\mathbb I(X_{ij}\le t)dF_{k}(t)}-2\sum_{j=1}^{c}\int{F_{j}(t)dF_{k}(t)}\right)_{1\le k\le c} 
\end{eqnarray}
and obtain
\begin{equation}\label{eq:4.7}
\frac{1}{\sigma_n}\sum_{i=1}^{n}\xi_{i}=\left(\frac{N(n)cB_{n}^{(1)}}{\sigma_n},...,\frac{N(n)cB_{n}^{(c)}}{\sigma_n}\right).
\end{equation}
In the following we shall obtain an almost sure central limit theorem for the
vector
$(\frac{N(n)c}{\sigma_n}T_{n}^{(1)},...,\frac{N(n)c}{\sigma_n}T_{n}^{(c)})$
that holds under the null and alternative hypothesis. Under the null
hypothesis we then show that the Kruskal-Wallis statistic is written as a function of $T_{n}^{(1)},...,T_{n}^{(c)}$. 
In order to show the almost sure central limit theorem for the vectors
$(\frac{N(n)c}{\sigma_n}T_{n}^{(1)},...,\frac{N(n)c}{\sigma_n}T_{n}^{(c)})$ it
is sufficient to show the almost sure central limit theorem for the vectors in
(\ref{eq:4.6}) and the statistics in (\ref{eq:4.7}) since  we may argue as in
the proof of Theorem \ref{theo:3.1}. In order to show the almost sure central
limit theorem  we need to check the assumptions in Lifschits' Theorem 3.1 (\cite{LI2}).\\
The first assumption in \cite{LI2} is to have a distributional convergence for the vectors $\xi_i$. We need to assume that the vectors $\xi_{i}$'s have a finite covariance matrix $\Sigma_{i}, i\ge 1$ such that
\begin{equation}\label{eq:4.8} 
\frac{\Sigma_{1}+...+\Sigma_{n}}{n}\rightarrow \Sigma \text{ as } n\rightarrow \infty,
\end{equation}
where $\Sigma$ is a $c \times c$ matrix. Note that this is essentially a condition on the dependencies of the coordinates of the independent random vectors. Also note that $E(\xi_{i})=0$ and since all coordinates of the random vectors $\xi_{i}$ are bounded, they satisfy the Lindeberg condition. 
Now, by the multivariate central limit theorem it follows that  
\begin{equation*}
\zeta_{n}:=\frac{1}{\sqrt{n}}\sum_{i=1}^{n}\xi_{i}\stackrel{D}{\rightarrow} N(0,\Sigma).
\end{equation*}
The second assumption is 
\begin{equation*}
b_{k}\le c_{1}\ln\left(\frac{\sigma_k}{\sigma_{k-1}}\right)=c_{1}\ln \sqrt{\frac{k}{k-1}}, \text{ for some constant } c_1>0
\end{equation*}
where $b_k$ appears in the formula of the empirical measures $Q_n=\frac{1}{\gamma_n}\sum_{k=1}^{n}b_{k}\delta_{\zeta_{k}}$ and $\gamma_n:=\sum_{k=1}^{n}b_{k}$. It can be verified (for example the mean value theorem can be used) that  
\begin{equation*}
\frac{1}{n} \leq 2\ln \sqrt{\frac{n}{n-1}},
\end{equation*}
so we can take $b_n=\frac{1}{n}$ and $\gamma_{n}\sim \ln (n)$.\\

The last assumption is that for some $\epsilon>0$ it holds that 
\begin{equation*}
\sup_{k} E(\ln_{+}\ln_{+} ||\zeta_{k}||)^{1+\epsilon}<\infty.
\end{equation*}
It is easy to see that 
\begin{equation*}
E(\ln_{+}\ln_{+} ||\zeta_{k}||)^2 \le E||\zeta_{k}||^2 =\frac{1}{k}\sum_{i=1}^{k}E(\xi_{i1}^2+...+\xi_{ic}^2)\le \frac{1}{k}kc(2c)^2=4c^3,
\end{equation*}
since the vectors $\xi_i$ are independent, have expectation zero and are
bounded ($|\xi_{ik}|\le 2c$ for every $i$ and $k$).  This shows that the
sequences  $(T_n^{(1)},...,T_n^{(c)})$ satisfies the almost sure central limit
theorem with limiting distribution function $G_X$ of some normal random
vector $X$.

It follows for a continuous function  $f:\mathbb{R}^{c}\to \mathbb{R}$  that
\begin{equation*}
\frac{1}{\ln n}\sum_{k=1}^{n}\frac{1}{k}\mathbb I\left( f(\frac{N(k)c}{\sigma_k}T_{k}^{(1)},...,\frac{N(k)c}{\sigma_k}T_{k}^{(c)}) \le \textbf{t}\right) \to G_{f(X)}(\textbf{t}), 
\end{equation*}
where $G_{f(X)}$ denotes the distribution function of $f(X)$.
In particular this applies to $f(\mathbf{x})=\mathbf{x}^{T}\mathbf{x}$.\\

It is left to show that under the null hypothesis the Kruskal-Wallis
statistics is such a function of $T_{n}^{(1)},....,T_{n}^{(c)}$. In this case, for a fixed index $l$ the model is described by:\\
\begin{align*}
H_n(t)&=F(t), &
\hat{H}_{n}(t)&=\frac{1}{N(n)}\sum_{i=1}^{n}\sum_{j=1}^{c}\mathbb I(X_{ij}\le t),\\
F_{n}^{(l)}(t)&=\frac{1}{c}F(t), &
\hat{F}_{n}^{(l)}(t)&=\frac{1}{N(n)}\sum_{i=1}^{n}\mathbb I(X_{il}\le t),
\end{align*}
\begin{equation*}
B_{n}^{(l)}(J)=\frac{1}{N(n)c}\sum_{i=1}^{n}\left(cF(X_{il})-\sum_{j=1}^{c}F(X_{ij})\right),
\end{equation*}
\begin{equation*}
T_{n}^{(l)}=L_{n}^{(l)}-\int{H_{n}(t)dF_{n}^{(l)}}=\frac{1}{N(n)(N(n)+1)}\sum_{i=1}^{n}R_{il}-\frac{1}{2c}.
\end{equation*}
Then, given that the Kruskal-Wallis statistic is defined by
\begin{equation}\label{eq:4.11}
F_{n}^{KW}=\frac{12}{N(n)(N(n)+1)}\frac{1}{n}\sum_{l=1}^{c}R_{l}^{2}-3(N(n)+1), \text{ where } R_{l}=\sum_{i=1}^{n}R_{il},
\end{equation}
it is possible to rewrite it as 
\begin{equation}\label{eq:4.12}
F_{n}^{KW}=\frac{12N(n)(N(n)+1)}{n}\sum_{k=1}^{c}\left(T_{n}^{(k)}\right)^2.
\end{equation}
Notice that under the null hypothesis, the asymptotic distribution of the
Kruskal-Wallis statistic calculated for $c$ dependent samples is not known. Still, an almost sure central limit theorem holds for the Kruskal-Wallis statistics. If (\ref{eq:4.8}) holds, then by the relationships (\ref{eq:4.11}) and (\ref{eq:4.12}) it follows that, 
\begin{equation}\label{eq:4.13}
\frac{1}{\ln n}\sum_{k=1}^{n}\frac{1}{k}\mathbb I\left(\frac{kc^3}{12(kc+1)}F_k^{KW}\le t\right) \rightarrow G_{X^{T}X}(t), 
\end{equation}
where $X\sim \mathcal{N}(0,\Sigma)$ and $t\in \mathbb{R}$.\\
\\
\begin{remark}\label{rem:4.1} {\rm Observe that in the classical case, when we consider random vectors with independent coordinates ($c$ independent samples), the Kruskal-Wallis statistic converges in distribution and an almost sure limit theorem holds under the null hypothesis.
\begin{equation}\label{eq:4.14}
F_{n}^{KW}\stackrel{D}{\longrightarrow} \chi_{c-1}^{2} \Rightarrow \frac{nc^3}{12(nc+1)}F_n^{KW}{\longrightarrow} \frac{c^2}{12}\chi_{c-1}^{2}
\end{equation}
\begin{equation*}
\frac{1}{\ln n}\sum_{k=1}^{n}\frac{1}{k}I\left(\frac{kc^3}{12(kc+1)}F_k^{KW}\le t\right) \rightarrow G_{X^{T}X}(t).
\end{equation*}
Here it is left to show that $X^{T}X\sim \frac{c^2}{12}\chi_{c-1}^{2}$. If we take the expectation in (\ref{eq:4.14}), then
\begin{equation*}
\frac{1}{\ln n}\sum_{k=1}^{n}\frac{1}{k}P\left(\frac{kc^3}{12(kc+1)}F_k^{KW}\le t\right) \rightarrow G_{X^{T}X}(t)
\end{equation*}
and using (\ref{eq:4.13}) it follows $G_{X^{T}X}=G_{\frac{c^2}{12}\chi_{c-1}^{2}}$.}\end{remark}

\begin{remark}\label{rem:4.2} 
{\rm Let us summarize the results that we obtained for the Kruskal-Wallis
  statistic. When considering $c$ dependent samples, the asymptotic
  distribution under the null hypothesis is not known but an almost sure limit
  theorem holds. Note that the limiting distribution of the almost sure
  central limit theorem exists but it cannot be calculated. From an applied
  point of view, the existence of the almost sure central limit theorem allows
  us to use the empirical quantile estimation method described in
  Remark \ref{rem:3.4} for statistical decisions. In the independent
  $c$-sample case, the limiting distribution of the weak convergence and of the
  almost sure central limit theorem exists and it can be calculated under the
  null hypothesis.\\

 The null hypothesis that we test is 
\begin{equation}\label{eq:4.16}
H_{0}: F_1=...=F_c \Rightarrow H_0: \int_{-\infty}^{\infty}F_{j}dF_{k}=\frac{1}{2} \text{ for every } j, k=1,2,...,c
\end{equation}
Using (\ref{eq:4.12}) it is not
hard to see that the Kruskal-Wallis test rejects $H_0$ if and only if 
\begin{equation*}
\frac{nc^3}{12(nc+1)}F_{n}^{KW} >\widehat{t}_{1-\alpha}^{(n)}.
\end{equation*}
This provides an asymptotic $\alpha$-level test for which the power tends to one
under the alternative.}
\end{remark}

\section{Simulation results}\label{sec:4}

In this section we study simulations for the three-sample problem in Section \ref{sec:3}. First, for three independent samples we perform simulation studies to show that the empirical logarithmic quantiles are good approximations of the asymptotic chi-squared quantiles of the Kruskal-Wallis statistic. Second, in case of three dependent samples we investigate by simulations the type I error and the power of the test given in (\ref{eq:4.16}).\\
The computation of the empirical logarithmic quantiles requires the following
adjustments. First, for better results we deleted the first five terms in the
sum of the empirical logarithmic distribution since their contribution is
dominating. Secondly, since the empirical logarithmic distribution for a
general statistic $T_n=T_{n}(X_1,...,X_n)$ is not symmetric and the rejection
or acceptance region might depend on the order of the observations, we considered a number of random permutations of the observations and calculated the quantities of the permuted sequence of independent vectors. Now the empirical logarithmic $\alpha$-quantiles can be computed by  
\begin{equation*}
\widehat{t}_{\alpha}^{(n)}=\frac{\sum_{i=1}^{\text{per}}\widehat{t}_{\alpha}^{*i,(n)}}{\text{per}},
\end{equation*}
where $\text{``per''}$ is the number of permutations that we want to consider and $\widehat{t}_{\alpha}^{*i,(N)}$ is the empirical logarithmic $\alpha$-quantile for permutation $i$ and is given by 
\begin{equation*}
\widehat{t}_{\alpha}^{*i,(n)}=\max\{t| \frac{1}{C_{n}}\sum_{k=1}^{n}\frac{1}{k}I(T_{k}^{*i}<t)\le \alpha\},
\end{equation*}
where $T_{k}^{*i}=T_k(X_{\tau_{i}(1)},...,X_{\tau_{i}(k)})$ and $\tau_{i}$ is the $i$-th permutation of $\{1, 2,..., n\}$. \\
\\
For three independent samples we use the almost sure central limit theorem given in (\ref{eq:4.14}). More precisely, using the method described above we compute the empirical logarithmic $\alpha$-quantile for the statistic $\frac{27n}{12(3n+1)}F_{n}^{KW}$ and compare it to the asymptotic $\alpha$-quantile given by $\frac{9}{12}\chi^{2}(2)$.\\
\\
For this purpose we run 500 simulations, consider 1000 observations in each
sample and permute independently each sample 100 times. We generate random
observations from a normal and an exponential distribution. The numerical
results of our simulation study are presented in Table 1. Note that all
quantiles are approximations to the unknown true distribution of the
statistics.The table indicates in particular, that the LQE-method is
distributionally stable.\\
\begin{table}
\centering
\begin{tabular}{l*{4}{c}r}
distribution/$\alpha$    & 1\% & 5\% & 10\%  \cr
\hline\hline
Normal(2,1)                    & 4.22799 & 3.43943 & 2.88437    \cr
Exponential (3)                & 4.23665 & 3.44352 & 2.88624    \cr
\hline\hline
$\frac{9}{12}\chi^{2}(2)$      & 6.907755 & 4.493598 & 3.453878 
\end{tabular}
\caption{Averaged empirical logarithmic $\alpha$-quantiles for three independent samples}
\end{table}

\noindent Also note that the simulation shows that the logarithmic quantiles
match the asymptotic quantiles sufficiently well. In order to decide which
approximation is better we approximated the true significance level associated to the two test procedures
$$
\frac{27\times 1000}{12(3\times 1000+1)}F_{1000}^{KW} >\widehat{t}_{0.9}^{(1000)}
$$
and
$$
\frac{27\times 1000}{12(3\times 1000+1)}F_{1000}^{KW} >\frac{9}{12}\chi^2(2)
$$
by counting the number of rejections among 500 simulations. It is found that,
for $\alpha=0.9$ (a test with significance level of 10\%), these covering probabilities are 0.108 for the LQE method and 0.12 for the asymptotic quantile method. This
shows that the LQE-method is (in this study) better suited than the classical method. 
\\
We now turn to the second simulation problem. In the case of three dependent samples the goal is to test whether the samples have the same distribution using the rejection rule in (\ref{eq:4.16}). We consider samples from a normal and an exponential distribution and calculate the type I error and the power for different levels of $\alpha$ and different sample sizes. Next we briefly describe the algorithms we used to simulate two dependent samples from normal and exponential distributions and to form three dependent samples we add one more independent sample for simplicity. Note that our simulation study is an indication of what can be expected from a detailed analysis of the performance of the test. \\
To generate two dependent samples from a normal distribution we will simulate dependent bivariate normal random variables with specific parameters and population correlation coefficient $\rho$. We use the algorithm from \cite{GEN}, page 197. We start by generating a matrix $\mathbf{X}_{(n\times 2)}$ of independent standard normal random variables (n independent copies of ($X_1, X_2$), where $X_1$ and $X_2$ are independent standard normal). Then we consider the covariance matrix $\Sigma$ of the vector $(X_1, X_2)$ given by (note that in this case the correlation coefficient $\rho_{X_1X_2}$is equal to $\sigma_{X_1X_2}$)
\begin{equation*}
 \Sigma= 
 \begin{pmatrix}
 1 & \sigma_{X_1X_2} \\
 \sigma_{X_1X_2} & 1
	\end{pmatrix}.
\end{equation*}
Using the Cholesky decomposition we obtain the matrix $\mathbf{T}_{(2\times 2)}$ such that $\mathbf{T}'\mathbf{T}=\Sigma$. Now $\mathbf{Y}=\mathbf{X}\mathbf{T}'$ gives n independent copies of bivariate vectors $(Y_1, Y_2)$, where $Y_1$ and $Y_2$ are dependent standard normal with correlation coefficient $\rho_{X_1X_2}$.\\     
To generate two dependent samples from an exponential distribution, we will simulate dependent bivariate exponential random variables with specific parameters and correlation coefficient. We use the Marshall-Olkin method described in \cite{DEV}, page 585. Start with the generation of three independent uniform random variables $U, V, S$ on $[0, 1]$ and then construct $X_1=\min\{-\frac{\ln U}{\lambda_1}, -\frac{\ln V}{\lambda_{3}}\}$ and $X_2=\min\{-\frac{\ln S}{\lambda_2}, -\frac{\ln V}{\lambda_{3}}\}$. In this way we obtain a bivariate vector $(X_1,X_2)$ with $X_1 \sim$Exp($\lambda_1+\lambda_3$), $X_2 \sim$Exp($\lambda_2+\lambda_3$) and the correlation coefficient is $\rho_{X_1X_2}=\frac{\lambda_3}{\lambda_{1}+\lambda_{2}+\lambda_{3}}$.\\
The results of our simulation studies are given in the tables 2--17. We start with the simulated significance level for different cases. \\ 

\begin{table}
\centering
\begin{tabular}{l||*{5}{c}r}
        & n=30 & n=50 & n=80 & n=100 & n=150 & n=200 \cr
\hline\hline
1\%     & 0 & 0 & 0 &  0 & 0 & 0    \cr
5\%    & 0.015 & 0.03 & 0.03 & 0.02 & 0.025 & 0.025     \cr
10\%   & 0.055 & 0.03 & 0.075 & 0.085  & 0.065 & 0.08       \cr
\end{tabular}
\caption{The level of significance for three $\mathcal{N}(0, 1)$ dependent samples; 200 simulations and 20 permutations, different sample sizes and different values of $\alpha$}
\end{table}

\begin{table}
\centering
\begin{tabular}{l||*{5}{c}r}
        & n=30 & n=50 & n=80 & n=100 & n=150 & n=200 \cr
\hline\hline
1\%     & 0.005 & 0 & 0 &  0 & 0 & 0    \cr
5\%    & 0.02 & 0.02 & 0.04 & 0.035 & 0.01 & 0.025     \cr
10\%   & 0.045 & 0.065 & 0.09 & 0.055  & 0.035 & 0.06      \cr
\end{tabular}
\caption{The level of significance for three Exp(4) dependent samples; 200 simulations and 20 permutations, different sample sizes and different values of $\alpha$}
\end{table}

\begin{table}
\centering
\begin{tabular}{l||*{5}{c}r}
        & n=30 & n=50 & n=80 & n=100 & n=150 & n=200 \cr
\hline\hline
1\%     & 0 & 0 & 0 &  0 & 0 & 0    \cr
5\%    & 0.025 & 0.04 & 0.05 & 0.05 & 0.045 & 0.04     \cr
10\%   & 0.065 & 0.1 & 0.095 & 0.08   & 0.125 & 0.125       \cr
\end{tabular}
\caption{The level of significance for three $\mathcal{N}(2, 1)$ independent samples; 200 simulations and 20 permutations, different sample sizes and different values of $\alpha$}
\end{table}

\begin{table}
\centering
\begin{tabular}{c||cccccc}
        & n=30 & n=50 & n=80 & n=100 & n=150 & n=200 \cr
\hline\hline
1\%     & 0 & 0 & 0 &  0 & 0 & 0    \cr
5\%    & 0.025 & 0.03 & 0.04 & 0.035 & 0.04 & 0.05     \cr
10\%   & 0.07 & 0.085 & 0.08 & 0.08   & 0.105 & 0.105       \cr
\end{tabular}
\caption{The level of significance for three Exp(3) independent samples; 200 simulations and 20 permutations, different sample sizes and different values of $\alpha$}
\end{table}

\noindent Note that the test is conservative and strongly conservative at $1\%$ level. A correction factor of 0.9 and a larger number of simulations will increase some significance levels. Also we noticed that 20 random permutations are seemingly sufficient.\\
We also compute the power of the test for different distributions, sample sizes and different significance levels.   

\begin{table}
\centering
\begin{tabular}{cccc||cccccc}
&$\mu_1$ & $\mu_2$ & $\mu_3$ & n=30 & n=50 & n=80 & n=100 & n=150 & n=200 \cr
\hline\hline	
& 0 & 1 & 0              & 1 & 1 & 1 & 1 & 1 & 1                       \cr
& 1 & 1 & 0              & 0.955 & 1 & 1 & 1 & 1  & 1                   \cr
& 0 & 0.5 & 0            & 0.44 & 0.83 & 0.975 & 0.995 & 1 &1            \cr
& 0 & 0.2 & 0            & 0.075 & 0.12 & 0.215 & 0.295 & 0.45 & 0.59     \cr
\end{tabular}
\caption{Power for three dependent samples from a normal distribution with different means at level $\alpha=10\%$; 200 simulations and 20 permutations}
\end{table}

\begin{table}
\centering
\begin{tabular}{cccc||cccccc}
&$\mu_1$ & $\mu_2$ & $\mu_3$ & n=30 & n=50 & n=80 & n=100 & n=150 & n=200 \cr
\hline\hline	
& 0 & 1 & 0              & 0.95 & 1 & 1 & 1 & 1 & 1                     \cr
& 1 & 1 & 0              & 0.825 & 0.97 & 1 & 1 & 1  & 1                 \cr
& 0 & 0.5 & 0            & 0.165 & 0.55 & 0.895 & 0.95 & 1 &1             \cr
& 0 & 0.2 & 0            & 0.015 & 0.04 & 0.08 & 0.08 & 0.205 & 0.38       \cr
\end{tabular}
\caption{Power for three dependent samples from a normal distribution with different means at level $\alpha=5\%$; 200 simulations and 20 permutations}
\end{table}

\begin{table}
\centering
\begin{tabular}{cccc||cccccc}
&$\mu_1$ & $\mu_2$ & $\mu_3$ & n=30 & n=50 & n=80 & n=100 & n=150 & n=200 \cr
\hline\hline	
& 0 & 1 & 0              & 0.155 & 0.81 & 1 & 1 & 1 & 1                     \cr
& 1 & 1 & 0              & 0.18 & 0.72 & 0.97 & 1 & 1  & 1                   \cr
& 0 & 0.5 & 0            & 0.005 & 0.025 & 0.1 & 0.195 & 0.615 & 0.835        \cr
& 0 & 0.2 & 0            & 0 & 0.005 & 0 & 0 & 0.01 & 0.025                   \cr
\end{tabular}
\caption{Power for three dependent samples from a normal distribution with different means at level $\alpha=1\%$; 200 simulations and 20 permutations}
\end{table}

\begin{table}
\centering
\begin{tabular}{cccc||cccccc}
&$\mu_1$ & $\mu_2$ & $\mu_3$ & n=30 & n=50 & n=80 & n=100 & n=150 & n=200 \cr
\hline\hline	
& 0.25 & 0.2 & 0.25    & 0.055 & 0.165 & 0.245 & 0.25 & 0.435 & 0.595           \cr
& 1 & 1 & 2            & 0.575 & 0.805 & 0.97 & 0.98 & 1  & 1                    \cr
& 0.5 & 1 & 0.5        & 0.625 & 0.89 & 0.995 & 1 & 1 & 1                         \cr
& 1 & 1 & 0.75         & 0.16 & 0.275 & 0.355 & 0.49 & 0.59 &  0.74                \cr
\end{tabular}
\caption{Power for three dependent samples from an exponential distribution with different means at level $\alpha=10\%$; 200 simulations and 20 permutations}
\end{table}

\begin{table}
\centering
\begin{tabular}{cccc||cccccc}
&$\mu_1$ & $\mu_2$ & $\mu_3$ & n=30 & n=50 & n=80 & n=100 & n=150 & n=200 \cr
\hline\hline	
& 0.25 & 0.2 & 0.25    & 0.035 & 0.04 & 0.085 & 0.145 & 0.235 & 0.405           \cr
& 1 & 1 & 2            & 0.355 & 0.665 & 0.915 & 0.955 & 0.995 & 1               \cr
& 0.5 & 1 & 0.5        & 0.38 & 0.715 & 0.985 & 1 & 1 & 1                         \cr
& 1 & 1 & 0.75         & 0.04 & 0.115 & 0.215 & 0.285 & 0.43 & 0.515               \cr
\end{tabular}
\caption{Power for three dependent samples from an exponential distribution with different means at level $\alpha=5\%$; 200 simulations and 20 permutations}
\end{table}

\begin{table}
\centering
\begin{tabular}{cccc||cccccc}
&$\mu_1$ & $\mu_2$ & $\mu_3$ & n=30 & n=50 & n=80 & n=100 & n=150 & n=200 \cr
\hline\hline	
& 0.25 & 0.2 & 0.25    & 0.005 & 0.005 & 0.02 & 0.01 & 0.065 & 0.105         \cr
& 1 & 1 & 2            & 0.24 & 0.4 & 0.34 & 0.36 & 0.655 & 0.945             \cr
& 0.5 & 1 & 0.5        & 0.01 & 0.085 & 0.295 & 0.535 & 0.785 & 0.985          \cr
& 1 & 1 & 0.75         & 0.005 & 0.01 & 0.015 & 0.015 & 0.05 & 0.135                 \cr
\end{tabular}
\caption{Power for three dependent samples from an exponential distribution with different means at level $\alpha=1\%$; 200 simulations and 20 permutations}
\end{table}

\begin{table}
\centering
\begin{tabular}{cccc||cccccc}
&$\mu_1$ & $\mu_2$ & $\mu_3$ & n=30 & n=50 & n=80 & n=100 & n=150 & n=200 \cr
\hline\hline	
& 0 & 1 & 0              & 0.975 & 1 & 1 & 1 & 1 & 1                       \cr
& 1 & 1 & 0              & 0.965 & 1 & 1 & 1 & 1  & 1                   \cr
& 0 & 0.5 & 0            & 0.445 & 0.775 & 0.94 & 0.975 & 1 & 1            \cr
& 0 & 0.2 & 0            & 0.11 & 0.22 & 0.26 & 0.335 & 0.485 & 0.605     \cr
\end{tabular}
\caption{Power for three independent samples from a normal distribution with different means at level $\alpha=10\%$; 200 simulations and 20 permutations}
\end{table}

\begin{table}
\centering
\begin{tabular}{cccc||cccccc}
&$\mu_1$ & $\mu_2$ & $\mu_3$ & n=30 & n=50 & n=80 & n=100 & n=150 & n=200 \cr
\hline\hline	
& 0 & 1 & 0              & 0.905 & 0.99 & 1 & 1 & 1 & 1                     \cr
& 1 & 1 & 0              & 0.885 & 1 & 1 & 1 & 1  & 1                        \cr
& 0 & 0.5 & 0            & 0.265 & 0.51 & 0.805 & 0.885 & 0.975 & 0.995       \cr
& 0 & 0.2 & 0            & 0.06 & 0.09 & 0.145 & 0.25 & 0.325 & 0.445           \cr
\end{tabular}
\caption{Power for three independent samples from a normal distribution with different means at level $\alpha=5\%$; 200 simulations and 20 permutations}
\end{table}

\begin{table}
\centering
\begin{tabular}{cccc||cccccc}
&$\mu_1$ & $\mu_2$ & $\mu_3$ & n=30 & n=50 & n=80 & n=100 & n=150 & n=200 \cr
\hline\hline	
& 0 & 1 & 0              & 0.185 & 0.71 & 0.905 & 0.935 & 1 & 1                 \cr
& 1 & 1 & 0              & 0.12 & 0.62 & 0.945 & 0.97 & 1 & 1                    \cr
& 0 & 0.5 & 0            & 0.005 & 0.045 & 0.145 & 0.19 & 0.58 & 0.76             \cr
& 0 & 0.2 & 0            & 0 & 0.005 & 0.005 & 0 & 0.005 & 0.14                    \cr
\end{tabular}
\caption{Power for three independent samples from a normal distribution with different means at level $\alpha=1\%$; 200 simulations and 20 permutations}
\end{table}

\begin{table}
\centering
\begin{tabular}{cccc||cccccc}
&$\mu_1$ & $\mu_2$ & $\mu_3$ & n=30 & n=50 & n=80 & n=100 & n=150 & n=200 \cr
\hline\hline	
& 0.25 & 0.2 & 0.25    & 0.095 & 0.25 & 0.305 & 0.385 & 0.46 & 0.635           \cr
& 1 & 1 & 2            & 0.6 & 0.91 & 0.985 & 0.99  & 1  & 1                    \cr
& 0.5 & 1 & 0.5        & 0.595 & 0.845 & 0.99 & 0.995 & 1 & 1                    \cr
& 1 & 1 & 0.75         & 0.22 & 0.245 & 0.37 & 0.45 & 0.69 &  0.755               \cr
\end{tabular}
\caption{Power for three independent samples from an exponential distribution with different means at level $\alpha=10\%$; 200 simulations and 20 permutations}
\end{table}

\begin{table}
\centering
\begin{tabular}{cccc||cccccc}
&$\mu_1$ & $\mu_2$ & $\mu_3$ & n=30 & n=50 & n=80 & n=100 & n=150 & n=200 \cr
\hline\hline	
& 0.25 & 0.2 & 0.25    & 0.045 & 0.1 & 0.165 & 0.195 & 0.335 & 0.44           \cr
& 1 & 1 & 2            & 0.475 & 0.715 & 0.93 & 0.96 & 0.995 & 1               \cr
& 0.5 & 1 & 0.5        & 0.39 & 0.695 & 0.93 & 0.99 & 1 & 1                         \cr
& 1 & 1 & 0.75         & 0.065 & 0.145 & 0.27 & 0.28 & 0.59 & 0.605               \cr
\end{tabular}
\caption{Power for three independent samples from an exponential distribution with different means at level $\alpha=5\%$; 200 simulations and 20 permutations}
\end{table}

\begin{table}
\centering
\begin{tabular}{cccc||cccccc}
&$\mu_1$ & $\mu_2$ & $\mu_3$ & n=30 & n=50 & n=80 & n=100 & n=150 & n=200 \cr
\hline\hline	
& 0.25 & 0.2 & 0.25    & 0 & 0.005 & 0.02 & 0.03 & 0.12 & 0.15         \cr
& 1 & 1 & 2            & 0.015 & 0.12 & 0.32 & 0.425 & 0.77 & 0.925             \cr
& 0.5 & 1 & 0.5        & 0.005 & 0.07 & 0.315 & 0.34 & 0.75 & 0.93          \cr
& 1 & 1 & 0.75         & 0 & 0.005 & 0.005 & 0.03 & 0.055 & 0.15                 \cr
\end{tabular}
\caption{Power for three independent samples from an exponential distribution with different means at level $\alpha=1\%$; 200 simulations and 20 permutations}
\end{table}

\noindent Note that the statistical power is close to $0$ when the difference between the means is negligible and the sample size is small. We can also observe that it increases when the sample size and the significance level increase.

\section{Auxiliary results on almost sure convergence for rank statistics}\label{sec:5}
In this section we collect the results needed for the proof of
Theorem \ref{theo:3.1}. We use the notations and the assumptions that were
introduced in Section \ref{sec:2}. The proofs in this section rely on the following facts:\\
1) The Borel-Cantelli lemma.\\
2) The estimation of variances of sums is done by estimating covariances of summands uniformly if they do not vanish, multiplied by the number of non-vanishing covariances.\\
We shall use these without further mentioning.\\ 
\begin{remark}\label{rem:6.1} {\rm We define 
\begin{eqnarray*}
\phi_{iluv}^{j}(s,t) &=&\lambda_{il}(J'(H_{j}(s))\mathbb I(t\leq s)-J'(H_{j}(s))F_{uv}(s)-\\
&-&\int J'(H_{j}(x))\mathbb I(t\leq x)dF_{il}(x)+\int J'(H_{j}(x))F_{uv}(x)dF_{il}(x)).
\end{eqnarray*}
We list two properties of the functions $\phi_{ijkl}^{n}$ that are used all over in the proofs of the following lemmas,
\begin{eqnarray*}
&&E(\phi_{iluv}^{j}(X_{il},X_{uv}))=0 \text{ if } i\neq u\\
&&E(\phi_{iluv}^{j}(X_{il}, X_{uv})\phi_{i'l'u'v'}^{j}(X_{i'l'}, X_{u'v'}))=0,
\end{eqnarray*}
if one of the indices $i, i', u, u'$  is different from the others, where $X_{il}$ is the $l$-th component of the vector $\mathbf{X}_{i}$.}\end{remark}
\begin{lemma}\label{lem:6.2}
\begin{equation*}
\frac{N(n)}{\sigma_{n}(J)}C_{1}(n)\longrightarrow 0  \text{ a.s. when }  n\longrightarrow\infty.
\end{equation*}
\end{lemma}
\begin{proof} We use the estimate of the second moment of $\frac{N(n)}{\sigma_{n}(J)}C_{1}(n)$ from \cite{BR1} 
\begin{equation}\label{eq:6.1}
E\left(\frac{N(n)}{\sigma_{n}(J)}C_{1}(n)\right)^{2}=O\left(\frac{N(n)^{2}|| J'||_{\infty}^{2}}{\sigma_{n}^{2}(J)N(n)^{4}}\left(\sum_{i=1}^{n}m_{i}^{2}\right)^{2}\right)=
O\left(\left(\frac{\max_{1\leq i\leq n}m_{i}}{\sigma_{n}(J)}\right)^{2}\right).
\end{equation} 
For the subsequence $n_{k}$ defined in Theorem \ref{theo:3.1}, for every $\epsilon >0$ and by Chebychev's inequality and relation (\ref{eq:6.1}) it follows 
\begin{equation}\label{eq:6.2}
P\left(\frac{N(n_{k})}{\sigma_{n_{k}}(J)}\left|C_{1}(n_{k})\right|>\epsilon\right) \le \frac{1}{\epsilon^2}E\left(\frac{N(n_{k})}{\sigma_{n_{k}}(J)}C_{1}(n_{k})\right)^2\le
\frac{1}{\epsilon^2}\left(\frac{\max_{1\leq i\leq n_k}m_{i}}{\sigma_{n_k}(J)}\right)^2.  
\end{equation}
Using assumption (\ref{eq:3.14}) and relation (\ref{eq:6.2}) it follows that 
\begin{equation*}
\sum_{k=1}^{\infty} P\left(\frac{N(n_{k})}{\sigma_{n_{k}}(J)}\left|C_{1}(n_{k})\right| >\epsilon \right)\le \sum_{k=1}^{\infty}\frac{1}{\epsilon^2}\left(\frac{\max_{1\leq i\leq n_k}m_{i}}{\sigma_{n_k}(J)}\right)^2 < \infty.
\end{equation*}
and by the Borel-Cantelli lemma, 
\begin{equation*}
\frac{N(n_{k})}{\sigma_{n_{k}}(J)}C_{1}(n_{k})\longrightarrow 0 \text{
a.s. when $k$ }  \rightarrow \infty.
\end{equation*}
So far we proved that Lemma \ref{lem:6.2} holds for the subsequence $n_{k}$. In order to prove that it holds for the whole sequence it is necessary to show that what happens between the subsequence points does not influence the convergence. Thus, let $j\in \mathbb{R}$ such that $n_{k}\le j < n_{k+1}$. The goal is to show that $\frac{N(j)}{\sigma_{j}(J)}C_1(j)$ converges to zero a.s. when $j\rightarrow \infty$. Notice that $C_{1}(j)$ defined in (\ref{eq:3.18}) can be decomposed as follows 
\begin{eqnarray*}
C_{1}(j)
&=&\frac{1}{N(j)^2}\sum_{i=1}^{j}\sum_{l=1}^{m_{i}}\sum_{u=1}^{j}\sum_{v=1}^{m_{u}}\lambda_{il}(J'(H_{j}(X_{il}))\mathbb
I(X_{uv} \leq X_{il})- \\
& &-J'(H_{j}(X_{il}))F_{uv}(X_{il})-\int J'(H_{j}(t))\mathbb I(X_{uv}\leq t)dF_{il}(t) \\
& &+\int J'(H_{j}(t))F_{uv}(t)dF_{il}(t)).
\end{eqnarray*}
Thus,
\begin{equation*}
C_{1}(j)=\frac{1}{N(j)^2}\sum_{i=1}^{j}\sum_{l=1}^{m_{i}}\sum_{u=1}^{j}\sum_{v=1}^{m_{u}}\phi_{iluv}^{j}(X_{il}, X_{uv})=\frac{1}{N(j)^2}\Phi_{1,1}^{j,j}(j) ,
\end{equation*}
where we defined  
\begin{eqnarray*}
\phi_{iluv}^{j}(s,t) &=&\lambda_{il}(J'(H_{j}(s))\mathbb I(t\leq s)-J'(H_{j}(s))F_{uv}(s)-\\
&-&\int J'(H_{j}(x))\mathbb I(t\leq x)dF_{il}(x)+\int J'(H_{j}(x))F_{uv}(x)dF_{il}(x))
\end{eqnarray*}
and
\begin{equation*}
\Phi_{a,c}^{b,d}(j)=\sum_{i=a}^{b}\sum_{l=1}^{m_i}\sum_{u=c}^{d}\sum_{v=1}^{m_u}\phi_{iluv}^{j}(X_{il}, X_{uv}), \text{ for }a,b,c,d \in \mathbb{N}, a<b, c<d.
\end{equation*}
The term $\frac{N(j)}{\sigma_{j}(J)}C_1(j)$ can be estimated as
\begin{equation*}
\left|\frac{N(j)}{\sigma_{j}(J)}C_1(j)\right|\leq \left|\frac{N(j)}{\sigma_{j}(J)}C_1(j)-\frac{N(n_{k})}{\sigma_{n_{k}}(J)}C_1(n_{k})\right|+
\left|\frac{N(n_{k})}{\sigma_{n_{k}}(J)}C_1(n_{k})\right|.
\end{equation*}
Since we showed that $\frac{N(n_{k})}{\sigma_{n_{k}}(J)}C_{1}(n_{k})\longrightarrow 0$ a.s. when $k\to \infty$, it is left to prove that
\begin{equation*}
\left|\frac{N(j)}{\sigma_{j}(J)}C_1(j)-\frac{N(n_{k})}{\sigma_{n_{k}}(J)}C_1(n_{k})\right| \longrightarrow 0 \text{ a.s. when }j\to \infty.
\end{equation*} 
Now
\begin{eqnarray*}
&&\frac{N(j)}{\sigma_{j}(J)}C_1(j)-\frac{N(n_{k})}{\sigma_{n_{k}}(J)}C_1(n_{k})=\frac{1}{\sigma_{j}(J)N(j)}\Phi_{1,1}^{j,j}(j)-\frac{1}{\sigma_{n_{k}}(J)N(n_{k})}
\Phi_{1,1}^{n_k,n_k}(n_k)=\\
&&=\frac{1}{\sigma_{j}(J)N(j)}\Phi_{1,1}^{n_k,n_k}(j)-\frac{1}{\sigma_{n_{k}}(J)N(n_{k})}\Phi_{1,1}^{n_k,n_k}(n_k)
+\frac{1}{\sigma_{j}(J)N(j)}\Phi_{1,n_k+1}^{n_k,j}(j)+\\
&&+\frac{1}{\sigma_{j}(J)N(j)}\Phi_{n_k+1,1}^{j,n_k}(j)
+\frac{1}{\sigma_{j}(J)N(j)}\Phi_{n_k+1,n_k+1}^{j,j}(j).
\end{eqnarray*}
Let
\begin{equation*}
\text{A}(j)=\frac{1}{\sigma_{j}(J)N(j)}\Phi_{1,1}^{n_k,n_k}(j)
-\frac{1}{\sigma_{n_{k}}(J)N(n_{k})}\Phi_{1,1}^{n_k,n_k}(n_k)=A_1(j)+A_2(j),
\end{equation*}
where
\begin{eqnarray*}
A_{1}(j)&=&\frac{1}{\sigma_{j}(J)N(j)}
\left(\Phi_{1,1}^{n_k,n_k}(j)-\Phi_{1,1}^{n_k,n_k}(n_k)\right),\\
A_{2}(j)&=&\left(\frac{1}{\sigma_{j}(J)N(j)}-\frac{1}{\sigma_{n_{k}}(J)N(n_{k})}\right)\Phi_{1,1}^{n_k,n_k}(n_k).
\end{eqnarray*}
We will show that $A_1(j)$ and $A_2(j)$ converge to zero a.s. using the Borel-Cantelli lemma. By Remark \ref{rem:6.1} we estimate the second moment of $A_{1}(j)$ as
\begin{equation*}
E(A_{1}^2(j))=O\left(\frac{||J'||^2}{\sigma^2_{j}(J)N(j)^2}\left(\frac{N(j)-N(n_{k})}{N(j)}\right)^2\left((\max_{1\leq i\leq n_{k}}m_{i})N(n_k)\right)^2\right).
\end{equation*}
Now, for some constant $C>0$,  
\begin{eqnarray*}
&&\sum_{k=1}^{\infty}\sum_{n_{k}\le j< n_{k+1}}P(|A_{1}(j)|\geq\epsilon)\leq\\
&&\sum_{k=1}^{\infty}\sum_{n_{k}\le j< n_{k+1}}\frac{1}{\epsilon^2}\frac{||J'||^2}{\sigma^2_{j}(J)N(j)^2}
\left((\max_{1\leq i\leq n_{k}}m_{i})N(n_k)\right)^2\left(\frac{N(j)-N(n_{k})}{N(j)}\right)^2\le\\
&&\le C\frac{||J'||^2}{\epsilon^2}\sum_{k=1}^{\infty}\left(\frac{\max_{1\leq i\leq n_{k}}m_{i}}{\sigma_{n_k}(J)}\right)^2< \infty,
\end{eqnarray*}
which follows from assumption (\ref{eq:3.14}) and the facts that 
$N(n_k)\le N(j)< N(n_{k+1}), \sigma_{n_{k}}(J)\le \sigma_{j}(J)< \sigma_{n_{k+1}}(J), N(n_{k+1})-N(n_k)\ge n_{k+1}-n_k$ and $k^{2}\le N(n_k)<(k+1)^2$.
Thus $A_1(j)\longrightarrow 0$ a.s. when $j\to \infty$.\\
For $A_2(j)$ notice that 
\begin{equation*}
|A_2(j)|=\left|\frac{\sigma_{n_{k}}(J)N(n_{k})}{\sigma_{j}(J)N(j)}-1 \right|\left|\frac{N(n_{k})}{\sigma_{n_{k}}(J)}C_{1}(n_{k})\right|\rightarrow 0 \text{ a.s. when }j\to \infty. 
\end{equation*} 
Let $B(j)=\frac{1}{\sigma_{j}(J)N(j)}\Phi_{1,n_k+1}^{n_k,j}(j)$. We will show that $B(j)\longrightarrow 0$ a.s. when $j\to \infty$ using Borel-Cantelli. An estimate for the second moment of $B(j)$ is obtained using the Remark \ref{rem:6.1} and assumption (\ref{eq:3.16})
\begin{eqnarray*}
&&E(B(j))^2 \leq\frac{16||J'||^2}{\sigma_{j}^{2}(J)N(j)^2}\left(\sum_{\lambda=n_{k}+1}^{j}m_{\lambda}^2\right)\left(\sum_{i=1}^{n_{k}}m_{i}^2\right)\\
&&\le K\frac{16||J'||^2}{\sigma_{j}^{2}(J)N(j)^2}(N(j)-N(n_{k}))(\max_{1\leq i \leq n_{k}}m_{i})^2 N(n_{k}).
\end{eqnarray*}
Now,
\begin{eqnarray*}
&&\sum_{k=1}^{\infty}\sum_{n_{k}\le j< n_{k+1}}P(|B(j)|\geq\epsilon)\le\\
&&\le \frac{16||J'||^2}{\epsilon^2}K\sum_{k=1}^{\infty}\left(\frac{\max_{1\leq i \leq n_{k}}m_{i}}{\sigma_{n_{k}}(J)}\right)^2
\frac{(N(n_{k+1})-N(n_{k}))^2}{N(n_{k})}\le\\
&&\le\frac{32^2||J'||^2}{\epsilon^2} K
\sum_{k=1}^{\infty}\left(\frac{\max_{1\leq i \leq n_{k}}m_{i}}{\sigma_{n_{k}}(J)}\right)^2<\infty,
\end{eqnarray*}
so $B(j)$ converges to zero a.s. when $j\to \infty$.\\
Let $C(j)=\frac{1}{\sigma_{j}(J)N(j)}\Phi_{n_k+1,1}^{j,n_k}(j)$. It can be
shown that $C(j)$ converges to zero a.s. when $j \to \infty$ in the same way
as for $B(j)$.\\  

Let $D(j)=\frac{1}{\sigma_{j}(J)N(j)}\Phi_{n_k+1,n_k+1}^{j,j}(j)$. Using the same techniques it can be shown that $D(j)$ converges to zero a.s. An estimate of the second moment of $D(j)$ is expressed as 
\begin{equation*}
E(D(j))^2=O\left(\frac{||J'||^2}{\sigma_{j}^{2}(J)N(j)^2}\left(\sum_{i=n_{k}+1}^{j}m_{i}^{2}\right)^2\right).
\end{equation*}
By the Borel-Cantelli lemma and the assumption (\ref{eq:3.16}) it follows that 
\begin{eqnarray*}
&&\sum_{k=1}^{\infty}\sum_{n_{k}\le j< n_{k+1}}P(|D(j)|\geq\epsilon)\leq \\
&&\leq\frac{||J'||^2}{\epsilon^2}\sum_{k=1}^{\infty}\sum_{n_{k}\le j< n_{k+1}}
\left(\frac{\max_{1\leq i \leq n_{k+1}}m_{i}}{\sigma_{n_{k}}(J)}\right)^2
\frac{\left(N(j)-N(n_{k})\right)^2\sigma_{n_{k}}^{2}(J)}{\sigma_{j}^{2}(J)N(j)^2}\le\\
&&\le \frac{||J'||^2}{\epsilon^2}\sum_{k=1}^{\infty}
\left(\frac{\max_{1\leq i \leq n_{k+1}}m_{i}}{\sigma_{n_{k}}(J)}\right)^2
\frac{((k+2)^2-k^2)^3}{k^4}<\infty,
\end{eqnarray*}
thus $D(j)$ converges to zero a.s. when $j \to \infty$.
\end{proof}

\begin{lemma}\label{lem:6.3}
\begin{equation*}
\frac{N(n)}{\sigma_{n}(J)}C_{2}(n)\longrightarrow 0 \text{ a.s. when }  n\longrightarrow\infty.
\end{equation*}
\end{lemma}
\begin{proof}
Recall that $C_{2}(n)$ was defined in (\ref{eq:3.19}). Then using the assumption (\ref{eq:3.14}) it follows that 
\begin{eqnarray*}
&&\left|\frac{N(n)}{\sigma_{n}(J)}C_{2}(n)\right|=\\
&&\left|\frac{1}{N(n)(N(n)+1)\sigma_{n}(J)}\sum_{i=1}^{n}\sum_{j=1}^{m_{i}}\sum_{k=1}^{n}\sum_{l=1}^{m_{k}}
J'(H_{n}(X_{ij}))\mathbb I(X_{kl}\leq X_{ij})\right|\le\\
&&\le||J'||\frac{N(n)}{(N(n)+1)\sigma_{n}(J)}\longrightarrow 0 \text{ a.s. when } n\longrightarrow\infty.
\end{eqnarray*}
\end{proof}
\begin{lemma}\label{lem:6.4}
\begin{equation*}
D_{n}=\sup_{t\in \mathbb{R}}\left|\frac{1}{N(n)}\sum_{k=1}^{n}\sum_{l=1}^{m_{k}}\left(\mathbb
I(X_{kl}\leq t)-F_{kl}(t)\right)\right|=O(a_n) \text{ a.s. }
\end{equation*}
\end{lemma}
\begin{proof} We shall use Singh's theorem and lemma (\cite{SIN}). Define 
\begin{equation}\label{eq:6.7}
Y_{k}(t)=\frac{1}{m_k}\sum_{l=1}^{m_{k}}\mathbb I(X_{kl}\leq t),
\end{equation}
\begin{equation}\label{eq:6.8}
F_{k}(t)=EY_{k}(t)=\frac{1}{m_k}\sum_{l=1}^{m_{k}}F_{kl}(t).
\end{equation}
Using (\ref{eq:6.7}) and (\ref{eq:6.8}), $D_n$ can be expressed as
\begin{equation*}
D_n=\sup_{t\in \mathbb{R}}\left|\sum_{k=1}^{n}w_k \left(Y_k(t)-F_k(t) \right)\right|,
\end{equation*}
where the weights are $w_k=\frac{m_k}{N(n)}$.\\

Using Singh's lemma (\cite{SIN}), for $a_n\ge \frac{1}{N(n)}\sqrt{\sum_{k=1}^{n} m_{k}^{2}}$ the following inequality holds 
\begin{equation*}
P(D_{n}\ge a_n)<\frac{4a_nN(n)^2}{\sum_{k=1}^{n} m_{k}^{2}}\exp\left\{-2\left(\frac{a_{n}^2N(n)^2}{\sum_{k=1}^{n} m_{k}^{2}}-1\right)\right\}.
\end{equation*}
Using Singh's theorem (\cite{SIN}), for any sequence $a_n\ge \frac{1}{N(n)}\sqrt{\sum_{k=1}^{n} m_{k}^{2}}$ such that $\sum_{n=1}^{\infty}\left\{\frac{a_nN(n)^2}
{\sum_{k=1}^{n}
m_{k}^{2}}\exp \left(-2\left(\frac{a_{n}^2N(n)^2}{\sum_{k=1}^{n}
m_{k}^{2}}\right)\right)\right\}<\infty$, it follows that
\begin{equation*}
D_{n}=O(a_n) \quad \text{ with probability 1}.
\end{equation*}
Take $a_n=b_{n}\frac{1}{N(n)}\sqrt{\sum m_{k}^{2}}$ with $b_n \sim c\sqrt{\log n}$ and $\sum_{n=1}^{\infty}\frac{b_n N(n)}{\sqrt{\sum_{k=1}^{n} m_{k}^{2}}}\exp(-2b_n^2) <\infty$. 
\end{proof}
\begin{lemma}\label{lem:6.5}:
\begin{equation*}
\frac{N(n)}{\sigma_{n}(J)}C_{3}(n)\longrightarrow 0 \text{ a.s. when }  n\longrightarrow\infty.
\end{equation*}
\end{lemma}
\begin{proof}
$C_{3}(n)$ defined in (\ref{eq:3.20}) can be written as 
\begin{eqnarray*}
C_{3}(n) &=&\frac{1}{2N(n)}\sum_{i=1}^{n}\sum_{j=1}^{m_{i}}\lambda_{ij}J''(\theta(H_{n}(X_{ij}))) \times\\
& \times &\left(\frac{1}{N(n)+1}\sum_{k=1}^{n}\sum_{l=1}^{m_{k}}\mathbb I(X_{kl}\leq X_{ij}) -\frac{1}{N(n)}\sum_{k=1}^{n}\sum_{l=1}^{m_{k}}F_{kl}(X_{ij})\right)^2.
\end{eqnarray*}
Thus,
\begin{eqnarray*}
&&\left|\frac{N(n)}{\sigma_{n}(J)}C_{3}(n)\right|= \frac{1}{2\sigma_{n}(J)}|\sum_{i=1}^{n}\sum_{j=1}^{m_{i}}\lambda_{ij}J''(\theta(H(X_{ij})))
(\frac{1}{N(n)+1}\sum_{k=1}^{n}\sum_{l=1}^{m_{k}}\mathbb \mathbb I(X_{kl}\leq X_{ij})-\\
&&-\frac{1}{N(n)}\sum_{k=1}^{n}\sum_{l=1}^{m_{k}}F_{kl}(X_{ij}))^2|\leq \frac{1}{2\sigma_{n}(J)}||J''||\sum_{i=1}^{n}\sum_{j=1}^{m_{i}}
(\frac{1}{N(n)+1}\sum_{k=1}^{n}\sum_{l=1}^{m_{k}}\mathbb I(X_{kl}\leq X_{ij})-\\
&&-\frac{1}{N(n)}\sum_{k=1}^{n}\sum_{l=1}^{m_{k}}F_{kl}(X_{ij}))^2\le \frac{1}{2\sigma_{n}(J)}||J''||\sum_{i=1}^{n}\sum_{j=1}^{m_{i}}
2(\frac{1}{(N(n)+1)^2}+\\
&&+(\frac{1}{N(n)}\sum_{k=1}^{n}\sum_{l=1}^{m_{k}}(\mathbb I(X_{kl}\leq X_{ij})-F_{kl}(X_{ij})))^2=\frac{1}{\sigma_{n}(J)}||J''||\frac{N(n)}{(N(n)+1)^2}+\\
&&+\frac{1}{\sigma_{n}(J)}||J''||\sum_{i=1}^{n}\sum_{j=1}^{m_{i}}(\frac{1}{N(n)}\sum_{k=1}^{n}\sum_{l=1}^{m_{k}}(\mathbb
I(X_{kl}\leq X_{ij})-F_{kl}(X_{ij})))^2.
\end{eqnarray*}
Using Lemma \ref{lem:6.4} it follows that
\begin{eqnarray*}
&&\left|\frac{N(n)}{\sigma_{n}(J)}C_{3}(n)\right|\le \frac{1}{\sigma_{n}(J)}||J''||\frac{N(n)}{(N(n)+1)^2}+\frac{1}{\sigma_{n}(J)}||J''||\sum_{i=1}^{n}\sum_{j=1}^{m_{i}}D_{n}^{2}\le\\
&&\le \frac{1}{\sigma_{n}(J)}||J''||\frac{N(n)}{(N(n)+1)^2}+\frac{1}{\sigma_{n}(J)}||J''||N(n)\frac{b_n^2\max m_k}{N(n)}\longrightarrow 0,
\end{eqnarray*}
since $\frac{b_n^2\max m_k}{\sigma_n(J)}\longrightarrow 0$ by (\ref{eq:3.15}) and choosing $b_n$ of order $\log n$.
\end{proof}

\begin{lemma}\label{lem:6.6}
\begin{equation*}
\lim_{n\to \infty}\frac{1}{\ln n}\sum_{k=1}^{n}\frac{1}{k}\mathbb I(\frac{N(k)}{\sigma_{k}(J)}B_{k}(J)\le t)=\Phi(t).
\end{equation*}
\end{lemma}
\begin{proof}
The term $B_{n}(J)$ defined in (\ref{eq:3.10}) can be expanded as 
\begin{eqnarray}
B_{n}(J)&=&\frac{1}{N(n)}\sum_{i=1}^{n}\sum_{j=1}^{m_{i}}\lambda_{ij}(J(H_{n}(X_{ij}))-\int J(H_{n}(x))dF_{ij}(x)+\notag\\
&+&\int J'(H_{n}(x))\widehat{H}_{n}(x)dF_{ij}(x)-\int J'(H_{n}(x))H_{n}(x)dF_{ij}(x))\label{eq:6.12}.
\end{eqnarray}
If we let
\begin{eqnarray*}
\sum_{i=1}^{n}\alpha_{i}&=&\sum_{i=1}^{n}\sum_{j=1}^{m_{i}}\lambda_{ij}\left(J(H_{n}(X_{ij}))-\int
J(H_{n}(x))dF_{ij}(x)\right)= \\
&=&\sum_{i=1}^{n}\sum_{j=1}^{m_{i}}\lambda_{ij}\left(J(H_{n}(X_{ij}))-E(J(H_{n}(X_{ij})))\right)
\end{eqnarray*}
and
\begin{eqnarray*}
\sum_{i=1}^{n}\beta_{i}&=&\sum_{i=1}^{n}\sum_{j=1}^{m_{i}}\lambda_{ij}\left(\int
J'(H_{n}(x))\widehat{H}_{n}(x)dF_{ij}(x)-\int
J'(H_{n}(x))H_{n}(x)dF_{ij}(x)\right)=\\
&=&\sum_{i=1}^{n}\sum_{j=1}^{m_{i}}\sum_{k=1}^{n}\sum_{l=1}^{m_{k}}\frac{1}{N(n)}\lambda_{ij}\int
J'(H_{n}(x))\left(I(X_{kl}\leq x)-F_{kl}(x)\right)dF_{ij}(x)= \\
&=&\sum_{k=1}^{n}\sum_{l=1}^{m_{k}}\sum_{i=1}^{n}\sum_{j=1}^{m_{i}}\frac{1}{N(n)}\lambda_{ij}\int
J'(H_{n}(x))\left(\mathbb I(X_{kl}\leq x)-F_{kl}(x)\right)dF_{ij}(x)=\\
&=&\sum_{i=1}^{n}\sum_{j=1}^{m_{i}}\sum_{k=1}^{n}\sum_{l=1}^{m_{k}}\frac{1}{N(n)}\lambda_{kl}\int
J'(H_{n}(x))\left(\mathbb I(X_{ij}\leq x)-F_{ij}(x)\right)dF_{kl}(x)= \\
&=&\sum_{i=1}^{n}\left(\sum_{j=1}^{m_{i}}\frac{1}{N(n)}\sum_{k=1}^{n}\sum_{l=1}^{m_{k}}\lambda_{kl}\int
J'(H_{n}(x))\left(\mathbb I(X_{ij}\leq x)-F_{ij}(x)\right)dF_{kl}(x)\right) 
\end{eqnarray*}
and
\begin{eqnarray*}
\xi_{i}&=&\sum_{j=1}^{m_{i}}(\lambda_{ij}\left(J(H_{n}(X_{ij}))-\int
J(H_{n}(x))dF_{ij}(x)\right)+ \\
&+&\frac{1}{N(n)}\sum_{k=1}^{n}\sum_{l=1}^{m_{k}}\lambda_{kl}\int
J'(H_{n}(x))\left(\mathbb I(X_{ij}\leq x)-F_{ij}(x)\right)dF_{kl}(x)) 
\end{eqnarray*}
then the relationship (\ref{eq:6.12}) can be rewritten as 
\begin{equation*}
N(n)B_{n}(J)=\sum_{i=1}^{n}\alpha_{i}+\sum_{i=1}^{n}\beta_{i}=\sum_{i=1}^{n}\xi_{i},
\end{equation*}
and
\begin{equation}\label{eq:6.16}
\frac{N(n)B_{n}(J)}{\sigma_n(J)}=\frac{1}{\sigma_n(J)}\sum_{i=1}^{n}\xi_{i}.
\end{equation}
The proof of the lemma follows from Theorem 1 of \cite{BER}. This result is an almost sure central limit theorem for an independent sequence of random variables. Since in (\ref{eq:6.16}) $\frac{N(n)B_{n}(J)}{\sigma_{n}(J)}$ is expressed as the partial sum of the independent random variables $\xi_i$, it is left to check the assumptions in their theorem. First, we need to have a convergence in distribution, which is given by Theorem 3.1 in \cite{BR1}
\begin{equation*}
\frac{N(n)}{\sigma_{n}(J)}B_{n}(J)=\frac{1}{\sigma_{n}(J)}\sum_{i=1}^{n}\xi_{i} \stackrel{D}{\rightarrow} N(0, 1). 
\end{equation*}
In Theorem 1 of \cite{BER} we put $f_{l}(x_1,..,x_l)=\frac{1}{\sigma_{l}(J)}\sum_{i=1}^{l}x_i$ and $f_{k,l}(x_1,..,x_{l-k})=\frac{1}{\sigma_{l}(J)}\sum_{i=1}^{l-k}x_i$ where $1\le k\le l$ and put $c_l=l^{\gamma}$ where $\gamma$ is as in (\ref{eq:3.13}). By Cauchy-Schwarz and using (\ref{eq:3.11}) we conclude that 
\begin{equation*}
E\left(\left|\frac{1}{\sigma_{l}(J)}\sum_{i=1}^{k}\xi_i \right|\wedge 1 \right)\le \frac{1}{M}\left(\frac{k}{l}\right)^{\gamma}.
\end{equation*}
Thus the theorem applies and we conclude that 
\begin{equation*}
\lim_{n\to\infty}\frac{1}{\ln n}\sum_{k\leq
n}\frac{1}{k}\mathbb I(\frac{1}{\sigma_{k}(J)}\sum_{i=1}^{k}\xi_{i}<x)=\Phi(x) \text{
a.s. for any }\ x\in \mathbb{R}.
\end{equation*}
\end{proof}

\end{document}